\documentclass[runningheads]{llncs}
\usepackage{amsmath,amssymb}
\let\doendproof\endproof
\renewcommand\endproof{~\hfill\qed\doendproof}
\usepackage{latexsym} 
\usepackage{epsfig}
\usepackage{graphicx}

\begin{document}

\title{Weak Dominance Drawings and \\
Linear Extension Diameter}

\author{Evgenios M.~Kornaropoulos\inst{1,2} and Ioannis G.~Tollis\inst{1,2}}

\institute{Department of Computer Science, University of Crete, Heraklion, Crete, Greece \\
P.O. Box 2208, Heraklion, Crete, GR-71409 Greece 
\and Institute of Computer Science, Foundation for Research and Technology-Hellas, \\
Vassilika Vouton, P.O. Box 1385, Heraklion, GR-71110 Greece \\
\texttt{\{kornarop,tollis\}@ics.forth.gr }}

\maketitle

\begin{abstract}
We introduce the problem of  Weak Dominance Drawing for general directed acyclic graphs and
we show the connection with the linear extension diameter of a partial order $P$. 
We present complexity results and bounds.
\end{abstract}


\section{Introduction}
Dominance drawings are an important tool  for visualizing planar $st$-graphs. This drawing method 
has many important aesthetic properties, including small number of bends, good vertex placement, and symmetry display \cite{Drawing:book,DiBattista:1992p9010}. 
Furthermore, it encapsulates the aspect of characterizing the transitive closure of the digraph
by means of a geometric dominance relation among the vertices. 
A dominance drawing $\Gamma$ of a planar $st$-graph $G$ is a drawing, such that for any two vertices $u$ and $v$ there is a 
directed path from $u$ to $v$ in $G$ if and only if $X(u)\leq X(v)$ and $Y(u) \leq Y(v)$ in $\Gamma$. 
An asymptotically optimal algorithm to obtain a two-dimensional dominance drawing of a bipartite graph whenever such a drawing exists was presented in \cite{Eades:1994p8711}. 
 
In this paper we introduce the problem of  \textit{Weak Dominance Drawing} by extending the concept
of dominance to general directed acyclic graphs. 
We show the connection of weak dominance drawings with the linear extension diameter of a partial order $P$. 
We also present complexity results and bounds.

\section{Problem Definition}

The notion of dominance drawing \textit{dimension} of a directed acyclic graph $G$ (denoted as $dim$($G$)) is defined  as the
value of the smallest $k$ for which a $k$-dimensional dominance drawing of $G$ can be obtained~\cite{Eades:1994p8711}. 
The family of  planar $st$-graphs is  a subclass of directed acyclic graphs for which a (2-dimensional) dominance drawing 
 can be efficiently computed in linear time \cite{Drawing:book,DiBattista:1992p9010}.  
This is possible because planar $st$-graphs have dimension 2.
If a graph $G$ has $dim(G)>2$ then there is at least one pair of vertices $u,v\in V$ such that  $X(u)\leq X(v)$ and $Y(u) \leq Y(v)$ in $\Gamma$, while neither $u$ can reach $v$, nor $v$ can reach $u$. 
Thus, a relaxed condition needs to be introduced in order to obtain a dominance drawing for any directed acyclic graph (dag). 
Consider two vertices $u$ and $v$ in a dag $G$ such that there is no path from $u$ to $v$ (or else vertex $v$ is unreachable from $u$) but $X(u)\leq X(v)$ and $Y(u) \leq Y(v)$ in $\Gamma$. 
Then the implied path $(u,v)$ in $\Gamma$ will be called a \textit{falsely implied path} (or simply fip) and will be included in the set of fips implied by drawing $\Gamma$. 
The number of falsely implied paths of drawing $\Gamma$ will be denoted by $fip$($\Gamma$). 
The existence of falsely implied paths is the trade off in drawing  graphs with $dim(G)>2$ in the two dimensional plane using the concept of dominance.

We  introduce the concept of \textit{weak dominance drawing} by relaxing the necessity of the existence of a path: that is, a dominance relation
between coordinates  of vertices in $\Gamma$ does not necessarily indicate the existence of a directed path in $G$, while if a path exists in $G$ then there is a dominance relation between vertex coordinates in $\Gamma$. 
More formally, a weak dominance drawing $\Gamma$ of a directed acyclic graph $G$ is constructed, such that for any two vertices $u$ and $v$ if there is a 
directed path from $u$ to $v$ in $G$ then $X(u)\leq X(v)$ and $Y(u) \leq Y(v)$ in $\Gamma$. 
The challenge in this problem is to minimize the number of fips, $fip$($\Gamma$).
A topological sorting of a dag $G=(V,E)$ is defined as a bijective function $t$ that maps each vertex of $G$ to a number in  $[1,|V|]$,
with the additional property that for every edge $(u,v)\in E$, $t(u)$ is smaller than $t(v)$.
The \textit{intersection} of two topological sortings $t_{X},t_{Y}$ of vertices in $V$ is the set $I=\{(u,v)| t_{X}(u)<t_{X}(v) ,  t_{Y}(u)<t_{Y}(v)\}$.
We can formulate the corresponding decision problem as follows:

\begin{flushleft}
\textbf{(WDD) WEAK DOMINANCE DRAWING FOR DIRECTED ACYCLIC GRAPHS}\\
INSTANCE: A directed acyclic graph $G=(V,E)$ and a positive integer $C$ such that $|$E$|\leq$$C$$\leq \frac{|V|(|V|-1)}{2}$.\\

QUESTION: Does there exist a collection of two topological sortings $t_{X},t_{Y}$   of $V$ such that their intersection has cardinality $C$ or less?\\
\end{flushleft}

A simple example for a crown graph $C_{3}$ is shown in Figure \ref{crown:cases}. 
This graph cannot be drawn so that $fip$($C_{3}$)$=0$, due to the fact that $dim(C_{3})=3$. 
As depicted in Figure \ref{crown:cases}, there is at least one path that is implied by the coordinate assignment of vertices and does not belong to $C_{3}$. 
From the above definition of the Weak Dominance drawing problem it can be easily derived that if a directed acyclic graph $G$ admits a dominance drawing, then $G$ admits a weak dominance drawing $\Gamma$ such that $fip(\Gamma)=0$.

\begin{figure}
\centering
\includegraphics[scale=0.26]{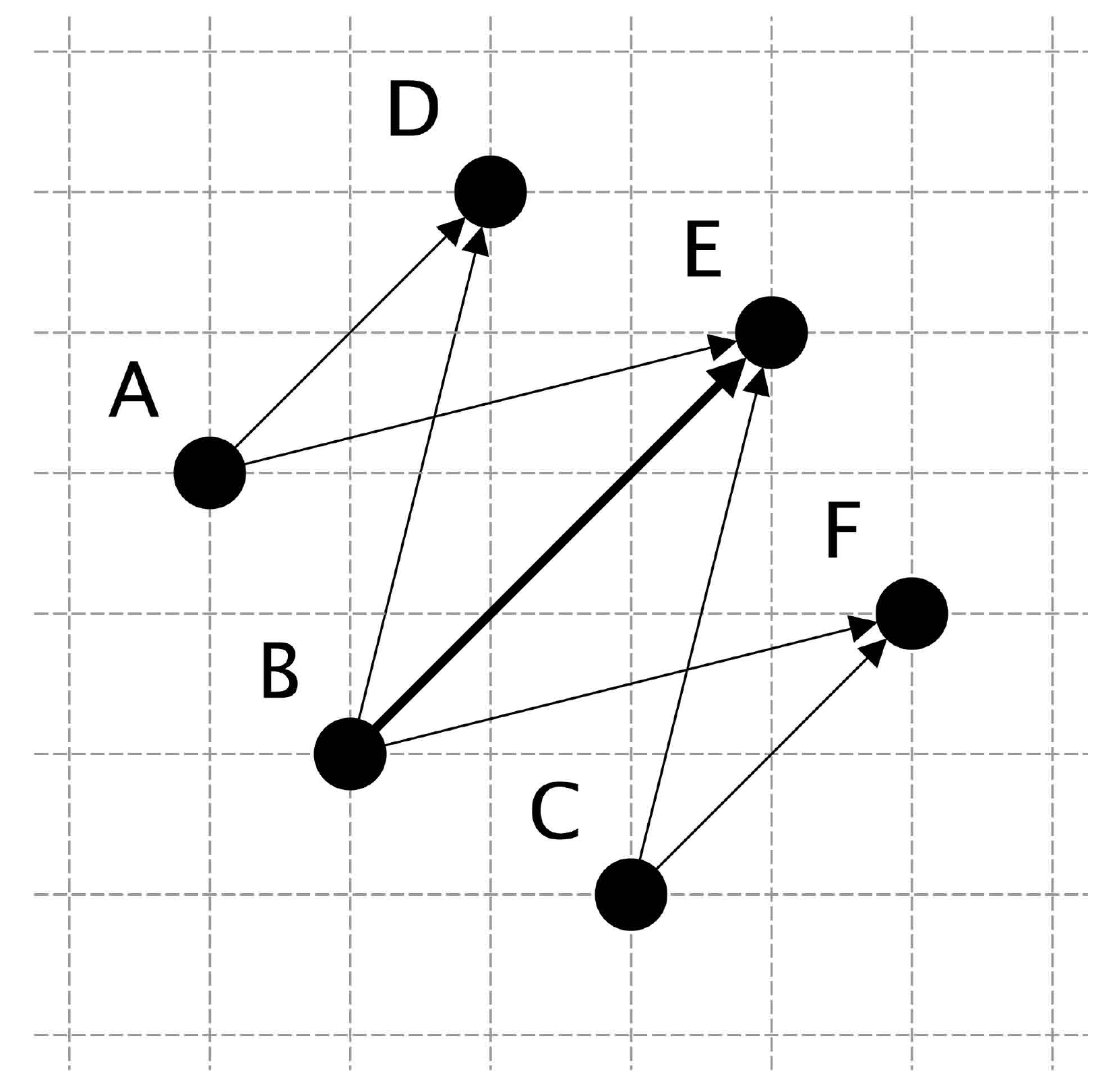}
\includegraphics[scale=0.26]{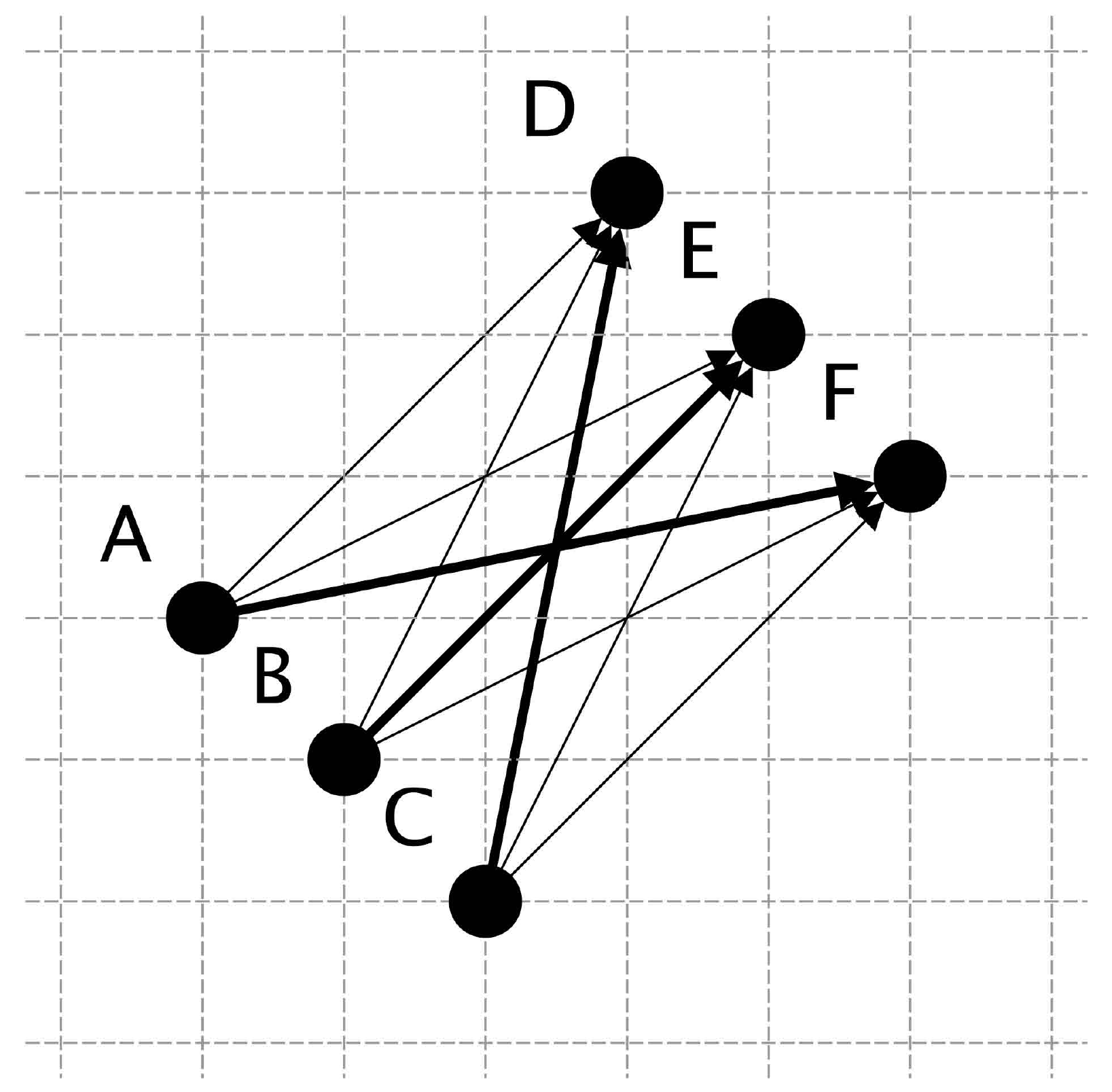}
\caption{Two different weak dominance drawings of crown  graph $C_{3}$.  Falsely implied paths appear with bold. In the left drawing there is an optimal weak dominance drawing with one falsely implied path, while in the right there is a weak dominance drawing with three falsely implied paths. } 
\label{crown:cases}
\end{figure}

An \textit{incomparable pair} of vertices $u,v$ has the property that neither $u$ is reachable from $v$, nor $v$ is reachable from $u$. We will denote as   $inc(G)$ the number of incomparable pairs of a dag $G$. 
The transitive closure of $G$ is the graph $G^{*}$=(V,$E^{*}$), where $E^{*}$=\{( $i$ , $j$ ): if there is a path from vertex $i$ to vertex $j$ in $G$\}. 
\begin{flushleft}
\textbf{Fact 1. }\textit{The number of incomparable pairs of a dag $G$ is:}
\begin{displaymath}
inc(G) = \frac{|V|(|V|-1)}{2}-|E^{*}|.
\end{displaymath}
\end{flushleft}

Clearly, according to the definition of fips, we have the following: 

\begin{flushleft}
\textbf{Fact 2. }\textit{Let $\Gamma$ be a weak dominance drawing of dag $G$. Then:}
\begin{displaymath}
fip(\Gamma) \leq inc(G).
\end{displaymath}
\end{flushleft}

The equality in the above bound is achieved for weak dominance drawings where all vertices are placed on the diagonal,
 i.e., $\forall u\in V, X(u)=Y(u)$. A general upper bound is given in the following:

\begin{lemma}If $G$ is a directed acyclic graph and $\Gamma$ is a weak dominance drawing of $G$, then the following inequality holds:
\begin{displaymath}
 min_{\Gamma}fip(\Gamma)  \leq inc(G) -(dim(G)-2).
\end{displaymath}
\end{lemma}

\begin{proof}
Let  $T=\{t_{1},t_{2},...,t_{d}\}$ be a set of topological sortings such that $d=dim$($G$). We choose two topological sortings $t_{i}, t_{j}$ from $T$  in order to construct a weak dominance drawing $\Gamma$, such that $\Gamma $ has the minimum number of fips comparing to any other pair of topological sortings in $T$. Since each topological sorting from  $T-\{t_{i},t_{j}\}$ decreases the number of fips in the intersection by at least one, drawing $\Gamma$ will have at most $inc(G)-(d-2)$ fips. Notice that although pair $t_{i},t_{j}$ is such that it yields the minimum number of fips out of the topological sortings in $T$, there could be another pair, not in $T$, that has an even lower number of fips. Thus, the inequality  $min_{\Gamma}fip(\Gamma)  \leq inc(G) -(dim(G)-2)$ holds.
 \end{proof}
 
 Another upper bound for the minimum number of fips for directed acyclic graphs is:
 
\begin{lemma}Let $G$ be a directed acyclic graph and $\Gamma$ be a weak dominance drawing of $G$, then the following inequality holds:
 \begin{displaymath}
 min_{\Gamma}fip(\Gamma)  \leq inc(G) -\lceil\frac{2inc(G)}{dim(G)}\rceil
\end{displaymath}
 \end{lemma}

 \begin{proof}
Let  $T=\{t_{1},t_{2},...,t_{d}\}$ be a set of topological sortings such that $d=dim(G)$. Let us randomly choose two $t_{i}, t_{j}$ from $T$ in order to construct weak dominance drawing $\Gamma$. The probability that an incomparable pair of vertices $u,v$ is expressed as a fip in $\Gamma$ is at least $(d-1)/ \binom{d}{2}$. Thus the expected number of fips is $\sum_{i=1}^{inc(G)}(d-1)/\binom{d}{2}=2inc(G)/d$. Therefore in an optimal drawing  $min_{\Gamma}fip(\Gamma) $ will be less than  $inc(G) -\lceil\frac{2inc(G)}{dim(G)}\rceil$. 
\end{proof}


\section{Linear Extension Diameter in Partially Ordered Sets}

In this section we present the LINEAR EXTENSION DIAMETER (LED) problem and we describe how it is 
related to the WDD problem. 
Due to the fact that the LED problem is proved to be NP-complete we will conclude that 
the WDD problem is also NP-complete.

A linear order (or total order) $L$ is a binary relation on some set $A$, with the properties of (a) antisymmetry, (b)  transitivity and also (c)  the property that every pairing of elements of $A$ must be related by $L$. 
A \textit{partially ordered set} (or poset) is a pair ($A$, $P$) where $A$ is a set of elements and $P$ is a reflexive, antisymmetric, and transitive binary relation on the elements of $A$. 
We call $A$ the ground set while $P$ is a partial order on $A$. 
Sometimes we refer to a poset simply as partial order $P$. 

A total order $L$ is a \textit{linear extension} of a partial order $P$ if, whenever $x \leq y$ in $P$ it also holds that $x \leq y$ in $L$. 
To put it briefly linear extensions are permutations that do not violate the given binary relations between pairs of  $P$. 
The distance between two linear extensions $L_{i},L_{j}$ of $P$, denoted by $dist(L_{i},L_{j})$, is the number of pairs of elements that are in opposite order between the two linear extensions. 
Given a poset $(A,\leq)$, a pair $L_{1},L_{2}$ of linear extensions is called a diametral pair if it maximizes the distance among all pairs of linear extensions of $P$. 
The maximal distance is called \textit{linear extension diameter of P} and is denoted by $led$($P$).

 It has been proved that every partial order $P$  is the intersection of a family of linear orders. 
The minimum number of linear orders required in such a representation of $P$ is the dimension of $P$, $dim(P)$~\cite{posets:DushnikMiler,Trotter:book}. 
The complexity of finding the dimension of a partial order was proved to be NP-complete by Yannakakis \cite{Yannak:poset}.

\begin{figure}
\includegraphics[scale=0.57]{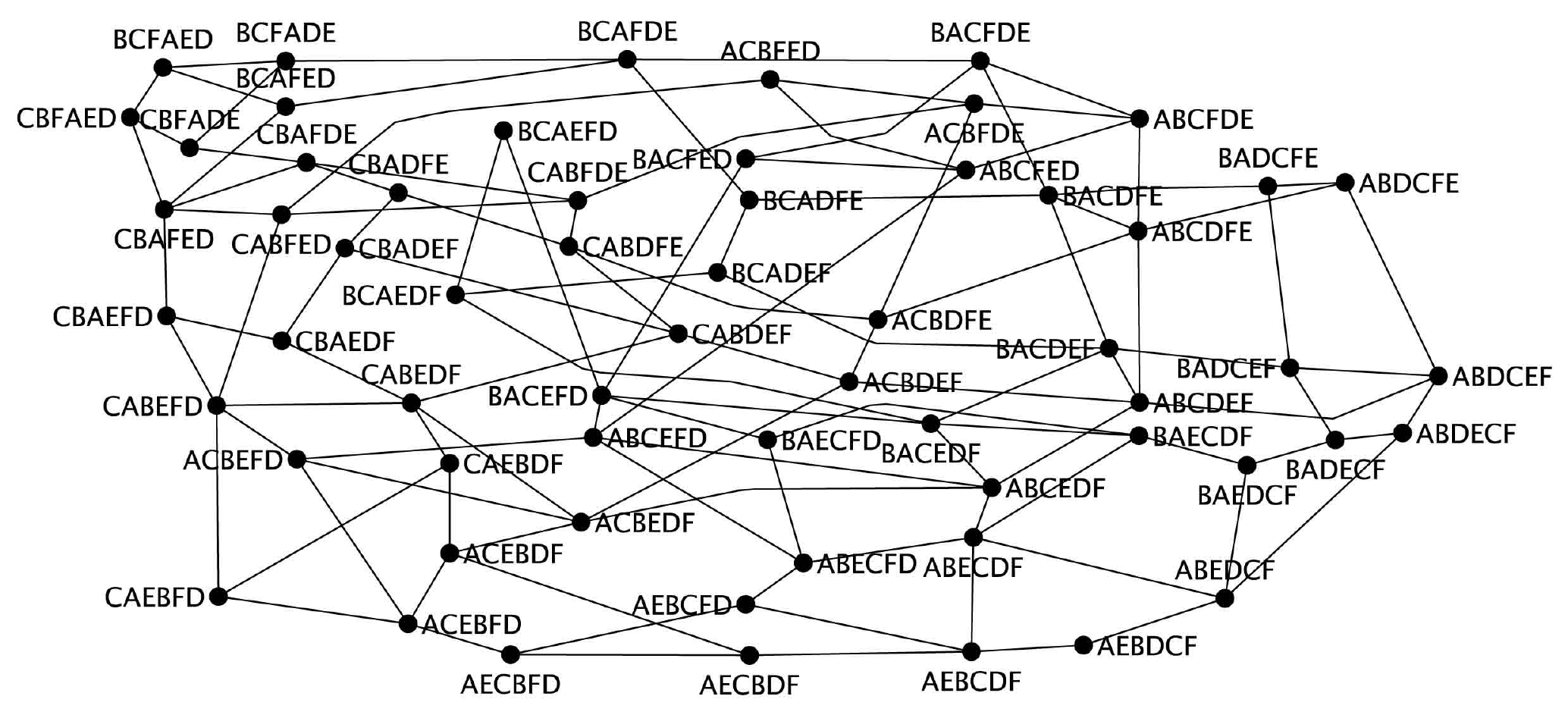}
\caption{The linear extension graph of the partial order derived from crown graph $C_{3}$. It has 61 vertices and 114 edges. A diametral pair is $L_{1}=ABDCEF$ and $L_{2}=CBFAED$ where $dist(L_{1}, L_{2})=8$ out of $inc(C_{3})=9$ incomparable pairs.}
\label{fig:example}
\end{figure}

The \textit{linear extension graph} $G(P)$ is the undirected graph that has a vertex for every different linear extension of $P$~\cite{Felsner:1999p14420}. 
Two vertices $L_{i},L_{j}$ of $G(P)$ are connected by an (undirected) edge if the linear extensions differ only by a single adjacent transposition. 
The problem of counting the number of vertices of a linear extension graph belongs to compexity class \#P, due to the fact that counting 
linear extensions also belongs to \#P \cite{Brightwell:1991p14245}. 
The distance between two linear extensions $L_{i},L_{j}$ of $P$, is equal to the shortest path distance between the corresponding vertices in $G(P)$~\cite{Naatz:led2000}.

If $P$ has dimension 2 there are two linear extensions $L_{i},L_{j}$ such that every pair $a,b$ of elements that cannot be compared, 
appear as $a\leq_{L_{i}}b$ in one linear extension and as $b\leq_{L_{j}}a$ in the other. 
Due to the fact that every incomparable pair of elements is in the opposite order between $L_{i}$ and $L_{j}$, the distance $dist(L_{i},L_{j})$
is equal to  the number of pairs of elements that cannot be compared. 
Since we have already considered every incomparable pair, this is also the maximal distance. 
Therefore, the linear extension diameter is equal to the number of pairs of elements that cannot be compared. 
The formulation of the LED problem is as follows:

\begin{flushleft}
\textbf{(LED) LINEAR EXTENSION DIAMETER  }\\
INSTANCE: A finite poset $(A,P)$, and a natural number $k$.\\
QUESTION: Are there two Linear Extensions of $P$ with distance at least $k$?\\
\end{flushleft}

In their work Brightwell and Massow \cite{Brightwell:2008p14422} proved that LED is NP-complete.

\subsection{Equivalence with Weak Dominance Drawing}

A partially ordered set can be viewed as a transitive directed acyclic graph with a set of nodes $A$ and an edge ($u,v$)$\in E$
for each pairwise relation between $u <_{P} v$ in $P$.
In these terms, a linear order is a complete dag. 
Linear extensions can be interpreted as topological sortings in graph theory. 
Therefore we can prove the following theorem:

\begin{theorem}
The LINEAR EXTENSION DIAMETER problem reduces to the WEAK DOMINANCE DRAWING FOR DIRECTED ACYCLIC GRAPHS problem.
\end{theorem}

\begin{proof}
In order to reduce the Linear Extension Diameter(LED) problem to the Weak Dominance Drawing(WDD) problem
we construct a dag $G$ from a given a partial order $P$ on the set of elements $A$, as described above.
For every element $u$ of $A$ we have a vertex $u$ in $G$, and for every relation $u$$<_{P}$$v$ between two elements,
we have a directed edge ($u,v$). 

First we prove that an optimal solution to LED is an optimal solution to WDD.
Let $L_{i},L_{j}$ be two linear extensions of the given partial order $P$,  that form a diametral pair of $P$, 
which implies that their distance is $dist(L_{i},L_{j})=k$. 
Thus the linear extension diameter of this partial order is $led(P)=k$. 
This means that out of $\binom{|A|}{2}$ possible pairs of elements of linear extension $L_{i}$, $k$ of them are in the opposite order in $L_{j}$.
Let us now choose these two linear extensions as the topological sortings for WEAK DOMINANCE DRAWING and denote them as $t_{i},t_{j}$. 
Then we will have $C=\binom{|A|}{2}-k$ pairs in the intersection, due to the fact that these $C$ pairs appear in the same order on both topological sortings. Since $k$ takes its maximum value for  partial order $P$ and $\binom{|A|}{2}$ is fixed, $C$ is minimum. Thus, the minimum cardinality of the intersection set is $C$.

Next, we prove that an optimal solution to WDD is an optimal solution to LED. 
Let us assume that  $t_{i}$ and $t_{j}$ are two topological sortings  such that their intersection has minimum cardinality $C$. 
Since there are $\binom{|A|}{2}$ possible pairs of vertices that may appear in the intersection, let $k$ be the number of pairs that do not appear 
in the intersection of $t_{i},t_{j}$. Then $C=\binom{|A|}{2}-k$. Since $C$ is minimum, it implies that the value of $k$ is maximum. 
Therefore, we will interpret $t_{i},t_{j}$ as linear extensions $L_{i},L_{j}$ in the LINEAR EXTENSION DIAMETER problem. 
Since we have $k$ pairs that appear in the opposite order between $t_{i}$ and $t_{j}$, we have $dist(L_{i},L_{j})=k$. 
Considering that $k$ is the maximum value that any two topological sortings can differ in their intersection, it is also the maximum distance
between any pair of linear extensions, therefore $led(P)=k$. 
Concluding, the LINEAR EXTENSION DIAMETER problem  has a solution if and only if the WEAK DOMINANCE DRAWING FOR DAGS
problem has a solution.
\end{proof}

\begin{corollary}
The WEAK DOMINANCE DRAWING FOR DIRECTED ACYCLIC GRAPHS  problem is NP-complete.
\end{corollary}

\section{Conclusion and Open Problems}

An interesting open problem is to better understand the interplay between the two topological sortings in order to minimize the number of falsely implied paths. 
In this direction approximation algorithms as well as heuristic algorithms should be developed.
Finally, the minimization of falsely implied paths within the intersection of more than two topological sortings is an interesting computational problem that is still open.

\end{document}